\documentclass[letterpaper]{article}

\usepackage{tda}
 \usepackage{orcidlink}
\usepackage{graphicx,amsmath,mathtools,psfrag}

\title{Optimistic Imprecise Shortest Watchtower in 1.5D and 2.5D}

\author{Bradley McCoy\thanks{Department of Computer Science, James Madison University, Harrisonburg, VA, \texttt{mccoy2ba@jmu.edu}}
	\and
	Binhai Zhu\thanks{Gianforte School of Computing, Montana State University, Bozeman, MT, \texttt{bhz@montana.edu}}}

\date{}
\begin{document}

\maketitle


\begin{abstract}

A 1.5D imprecise terrain is an $x$-monotone polyline with fixed $x$-coordinates, the $y$-coordinate of each vertex is not fixed but is constrained to be in a given vertical interval. A 2.5D imprecise terrain is a triangulation with fixed $x$ and $y$-coordinates, but the $z$-coordinate of each vertex is constrained to a given vertical interval. Given an imprecise terrain with $n$ intervals, the optimistic shortest watchtower problem asks for a terrain $T$ realized by a precise point in each vertical interval such that the height of the shortest vertical line segment whose lower endpoint
lies on $T$ and upper endpoint sees the entire terrain is minimized.
In this paper, we present a linear time algorithm to solve the 1.5D optimistic shortest watchtower problem exactly. For the discrete version of the 2.5D case (where the watchtower must be placed on a vertex of $T$), we give an additive approximation scheme running in $O(\log(\frac{{OPT}}{\varepsilon})n^3)$ time, achieving a solution within an additive error of $\varepsilon$ from the optimal solution value ${OPT}$.

\end{abstract}


\section{Introduction}
\label{sec:intro}

Placing watchtowers on polyhedral terrains is a fundamental problem in computational geometry.
In practice, terrains are often constructed by surveying aircraft, which record distances to the ground.
Since these measurements inevitably contain errors, the resulting
data typically specifies not precise heights but intervals of possible elevations.
This motivates the imprecise watchtower problem, a natural and practical extension of its classical counterpart.

Imprecise terrains were introduced by Gray and Evans \cite{gray-shortest-hard-04}, who
showed that computing an optimistic shortest path on a 2.5D imprecise terrain is NP-hard. Since then, the problem of finding realizations of imprecise terrains with desirable properties has developed into an active research area.

For a given imprecise terrain and geometric problem, one can ask
for a realization that yields the best (optimistic) possible outcome, or one that yields the worst (pessimistic) outcome.
Examples include approximation algorithms for minimizing maximum turning angles in 1.5D\cite{gray_smoothing_2010}, 
minimizing the number of extrema in 2.5D  \cite{gray_removing_2012}, water flow on 2.5D  terrains \cite{driemel_flow_2011}, and optimizing coplanar features, minimizing surface area, and minimizing maximum steepness \cite{Lubiw_Stroud_2023}.

For precise terrains, guarding with a single shortest watchtower has a rich history. In 1988, Sharir \cite{sharir_shortest_1988} gave an $O(n\log^2(n))$ time algorithm for guarding a 2.5D terrain with $n$ vertices, later
improved to $O(n\log n)$ by Zhu \cite{zhu_computing_1997}.
A shortest watchtower for a precise 1.5D terrain can be computed in 
linear time using the algorithm by Lee and Preparata \cite{lee_optimal_1979}.
There are two natural versions of the watchtower problem: the discrete version, where the watchtower base must be on a vertex of the terrain, and the continuous version, where it may be placed anywhere on the terrain. The results mentioned above are for the continuous version.
Several other variations of guarding a precise terrain have been studied
\cite{benmoshe_constantfactor_2007,CS89,chen_guarding_2018,agarwal_guarding_2010,seth_acrophobic_2023,king_terrain_2011}, and more recently, maximizing visibility on a 1.5D terrain with imprecise viewpoints was considered in \cite{Loffler_viewpoints_2025}. 
 
For imprecise watchtowers, McCoy, Zhu, and Dutt \cite{mccoy_guarding_2024} gave an additive PTAS for the optimistic 1.5D imprecise shortest watchtower problem.
In this paper, we strengthen that result by presenting an exact linear time algorithm for the 1.5D case.
We then extend our result to the discrete version of the problem in 2.5D, for which we give an additive approximation scheme running in $O(\log(\frac{{OPT}}{\varepsilon})n^3)$ time, where ${OPT}$ is the optimal solution value of the problem.

This paper is organized as follows.  \secref{preliminaries} provides necessary definitions. \secref{algorithm} presents our linear time algorithm for the 1.5D case. \secref{ptas} gives details for the 2.5D case. We conclude in \secref{discuss} with a discussion of related open problems.

\section{Preliminaries}
\label{sec:preliminaries}

 A \emph{1.5D terrain} is an $x$-monotone polyline with $n$ vertices.
 An \emph{imprecise 1.5D terrain} is a 1.5D terrain where each terrain vertex comes with a closed vertical interval rather than a fixed $y$-coordinate. We denote these $n$ vertical intervals
$\ell_1,\ell_2,\ldots, \ell_n$.
A \emph{realization} of an imprecise terrain is obtained by selecting a $y$-coordinate from each interval, thereby defining a point on every $\ell_i$ (see \figref{1.5D_with_polygon}). 
By slight abuse of notation, we refer to the polyline resulting from a realization as having \emph{vertices} and \emph{edges}, as in the precise case.
For each interval $\ell_i,$ let $v_i$ denote the chosen vertex, and let $t_i$ and $b_i$ denote the top and bottom endpoints of $\ell_i$, respectively.
A vertex is called \emph{left-turning} if, when the polyline is traversed from left to right, it makes a left (upward) turn at that vertex. Similarly, it is  \emph{right-turning} if the polyline turns right (downward),
and \emph{straight} if it does not turn.

Given a realization of an imprecise terrain $T$, 
consider the halfplanes defined by extending each edge of $T$
and selecting the side that contains the point $(0,+\infty)$.
The intersection of these halfplanes forms an unbounded convex polygon $P,$  called the \emph{visibility polygon}.
This polygon can be computed in linear time, has $O(n)$ vertices
 \cite{lee_optimal_1979} an example is the  shaded region of \figref{1.5D_with_polygon}.
 
A shortest watchtower for $T$ is a vertical line segment $\overline{uw}$, where the lower endpoint $u$ lies on $T$ and
the upper endpoint $w$ lies on $P.$
Two points $p$ and $q$ are \emph{visible} if the segment $\overline{pq}$ is above $T$. 
If every point in an edge is visible to a point we
say the edge is visible, similarly for faces.
If every point of $T$ is visible from a point $w,$ then we say
that $T$ is \emph{guarded} by $w.$
The \emph{optimistic shortest watchtower problem} asks for a realization $T$ that minimizes the length $|uw|$ of the shortest watchtower, where $u\in T$ and $w$ sees all of $T.$
The \emph{total length} of a realization $T$ is $$|T|=\sum_{i=1}^{n-1}||v_iv_{i+1}||.$$

\begin{figure}[t]
    \captionsetup[subfigure]{justification=centering}
    \centering
    \begin{subfigure}[t]{0.47\textwidth}
    	\includegraphics[width=\textwidth]{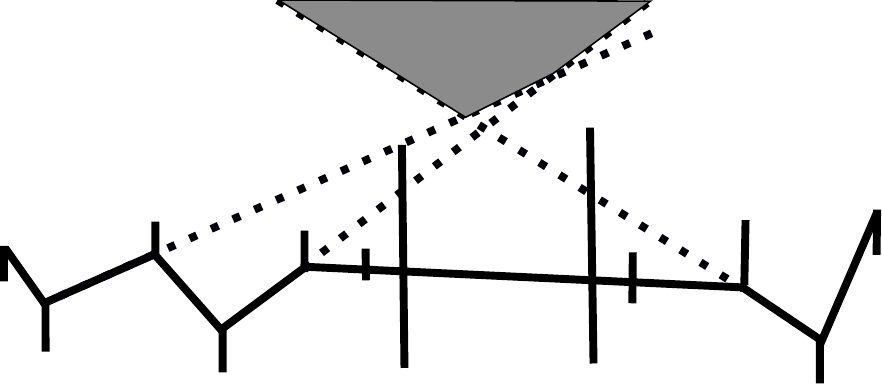}
	    \subcaption{}
    \label{fig:1.5D_with_polygon}
    \end{subfigure}
    \hspace{.75cm}
    \begin{subfigure}[t]{0.37\textwidth}
  	    \includegraphics[width=\textwidth]{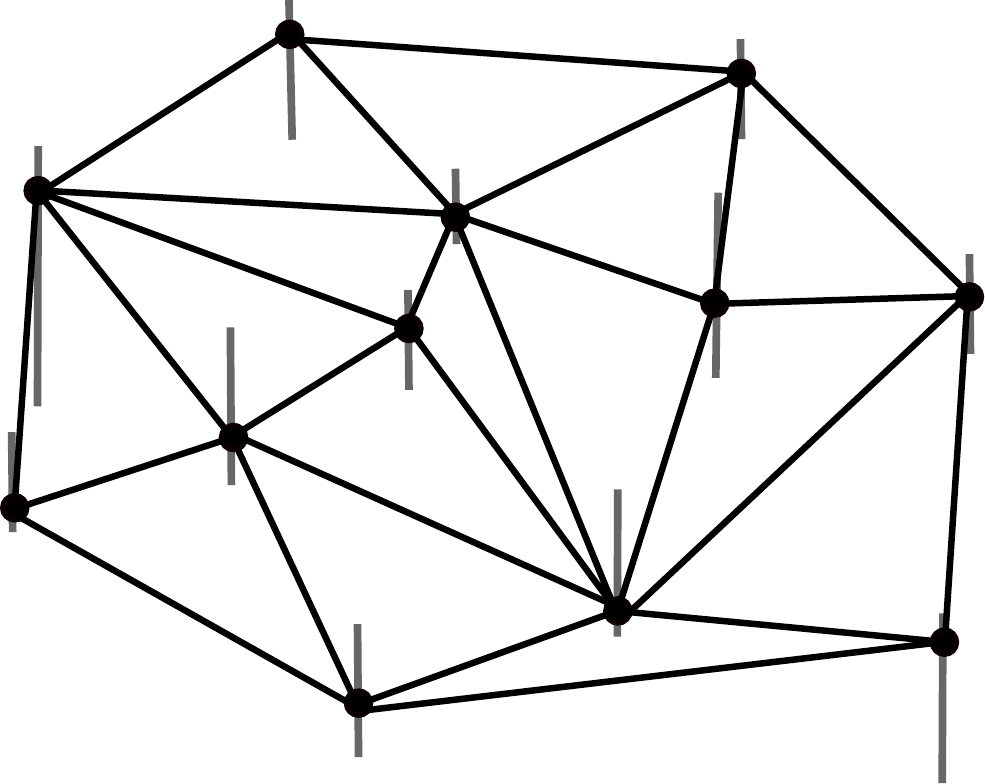}
	    \subcaption{}
    \label{fig:2.5D_imprecise}
    \end{subfigure}

	\caption{(\subref{fig:1.5D_with_polygon}) A possible realization of a 1.5D imprecise terrain. The shaded region represents the unbounded convex polygon $P$. (\subref{fig:2.5D_imprecise}) A possible realization of a 2.5D imprecise terrain.}
\label{fig:intro}
\end{figure}

Let $Q$ be the simple polygon formed as follows:
take the top and bottom endpoints of the intervals $\ell_i$
as vertices, connect consecutive top endpoints and consecutive bottom endpoints to form two chains, and include the vertical edges $\ell_1$ and $\ell_n.$ 
The collection of shortest paths from a fixed vertex to
every other vertex in a polygon can be computed in
linear time \cite{guibas_linear-time_1987}.
Let a shortest path $\pi$ from a point $s$ to a point $t$ in $Q$, the 
vertices of $\pi \setminus \{s,t\}$ are called \emph{inner vertices}.

In the \emph{imprecise 2.5D terrain model}, we are given $n$ fixed $x,y$-coordinates together with a triangulation defined 
by projecting the vertices onto the $xy$ plane. The $z$-coordinate of each vertex is specified only up to a given interval (\figref{2.5D_imprecise}). A \emph{realization} of an imprecise 2.5D terrain is obtained by selecting a $z$-coordinate within the allowed interval for each vertex. 


\section{Algorithm for the 1.5D Case}
\label{sec:algorithm}

In this section, we present an $O(n)$ time algorithm to compute an optimistic shortest watchtower on an imprecise 1.5D terrain. We show that any realization can be transformed into one of $O(n)$ canonical realizations without increasing the height of the watchtower.  We first consider the discrete version then extend the argument to the continuous version. The following lemmas apply to both cases.

\begin{lemma}[Raise Wings]\label{lem:wings}
Let $W$ be a watchtower for a realization $T$, and let $T'$ be the realization obtained from $T$ by raising $v_1$ and $v_n$ to the tops of their respective intervals. Then $W$ is also a watchtower for $T'$.
\end{lemma}

\begin{proof}

Let $w$ denote the top of the watchtower. Consider the quadrilateral $D$ formed by $w, v_2, v_1,$ and $t_1$. Since $w$ guards the edge $v_1v_2$, the interior angle at $v_2$ is at most $180^{\circ}$. An example is shown in \figref{wings}. Hence, the quadrilateral $D$ is convex, and the watchtower also sees the edge $t_1v_2$. Thus, $W$ is a valid watchtower for $T'$. The argument is symmetric for raising $v_n$ to $t_n$.
\end{proof}

\begin{figure}[h]
\centering
\includegraphics[width=0.4\textwidth]{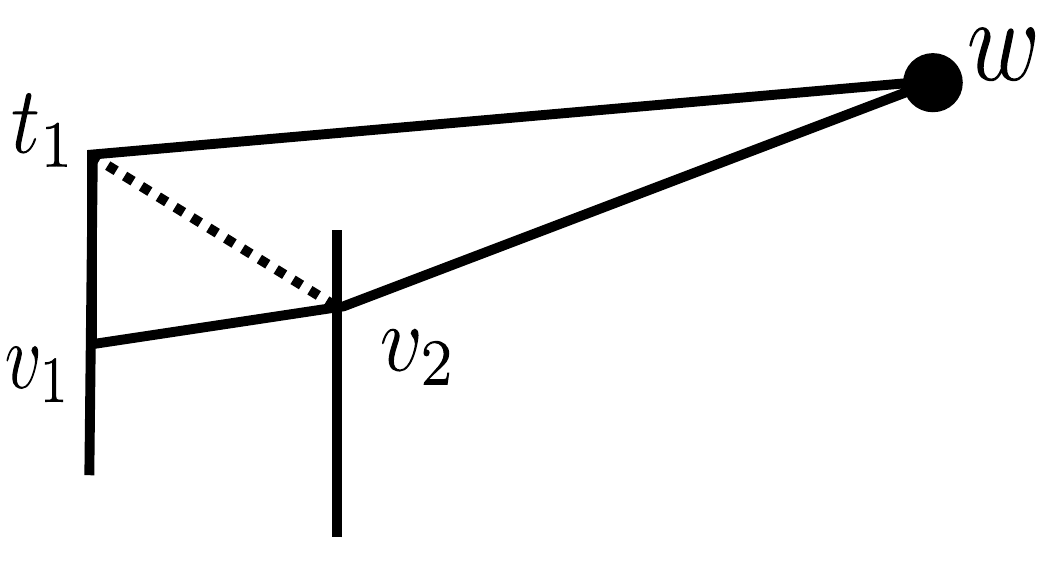}
\caption{The quadrilateral from the proof of \lemref{wings}.}
\label{fig:wings}
\end{figure}


Consider three consecutive vertices $v_{k-1},v_k,v_{k+1}$.
Define the \emph{local shortening operation} by replacing $v_k$
with a new point $v_k'\in \ell_k$ that minimizes
$||v_{k-1}v_k'||+||v_k'v_{k+1}||.$
Given a subpath
$P=(v_i,v_{i+1},\ldots,v_j),$
a \emph{local shortening of $P$} is the application of the above
operation to any interior vertex $v_k$ with $i<k<j$.  
Let $v_i=(x_i,y_i)$ and $v_j=(x_j,y_j)$ be two fixed vertices of a
realization with $x_i<x_j$. Consider any watchtower whose
$x$-coordinate lies outside the interval $[x_i,x_j]$.

\begin{lemma}[Shorten]\label{lem:shorten}
Applying a local shortening operation to any interior vertex of the
subpath $v_i\rightarrow v_j$ does not increase the height of the
watchtower. And, any sequence of local shortening operations
on the subpath $v_i\rightarrow v_j$ does not increase the watchtower
height.
\end{lemma}

\begin{proof}
Let $w$ denote the top of an optimal watchtower and suppose a local
shortening operation is applied at an interior vertex $v_k$,
$i<k<j$. Since the endpoints $v_i$ and $v_j$ remain fixed, only the
position of $v_k$ changes.

Consider the quadrilateral $D$ formed by
$v_{k-1},v_k,v_{k+1},w$. If $D$ is convex, it remains convex after
the local shortening operation. Otherwise, the interior angle at  $v_{k-1}v_kv_{k+1}$ is reflex. Locally shortening at $v_k$ decreases this angle, See \figref{shorten} for an example.  
Thus, if the edges $v_{k-1}v_k$ and $v_kv_{k+1}$ are visible from the watchtower before shortening, they remain visible after the path is shortened.

Thus a single local shortening operation does not increase the
required watchtower height. Since each shortening operation weakly
decreases the length of the subpath and strictly decreases it unless
the subpath is already locally shortest, repeated local shortening
converges to the unique shortest-path realization with the same fixed
endpoints. Because every step preserves or decreases the watchtower
height, transforming the subpath into its shortest-path realization
cannot increase the height of the watchtower.
\end{proof}

\begin{figure}[h]
\centering
\includegraphics[width=0.4\textwidth]{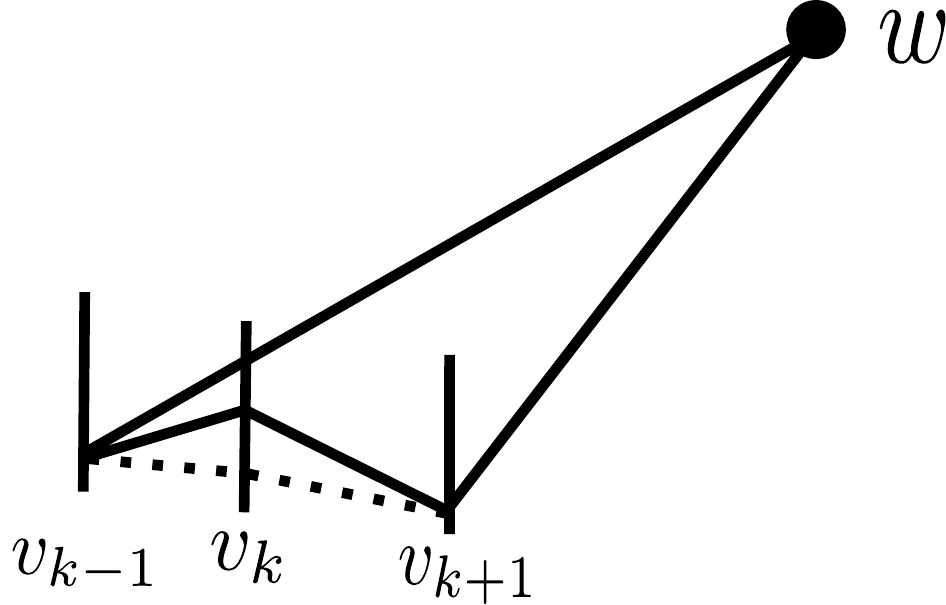}
\caption{Shortening a path with fixed endpoints does not increase
the height of a watchtower.}
\label{fig:shorten}
\end{figure}

The top of an optimal watchtower is determined by a vertex on $P$ or by an edge on $P$ and a vertex on $T.$ We now show that the edges (or edge) whose extensions determine the top of an optimal watchtower 
must belong to the realization given by a shortest path $\pi : t_1 \rightarrow t_n$.

 \begin{lemma}[Top of tower determined by $\pi$]\label{lem:half-planes-on-pi}
Let $\pi$ be a shortest path $\pi : t_1 \rightarrow t_n,$
let $T$ be any realization with optimal (shortest) watchtower height, and let $e$
denote an edge whose extension intersects at the top of the watchtower. 
Then $e$ is an edge of $\pi$.
\end{lemma}

\begin{proof}
Assume $e$ is not on $\pi$ and is to the left of the watchtower.
Let $E$ denote the maximal segment containing $e$ as a subsegment and let $L$ and $R$ denote the left and right endpoints of $E.$
See \figref{optimal-edge} for an example.
Since $e$ is not an edge on $\pi,$ either the $y$-value of $R$ is above $\pi$ or the $y$-value of $L$ is below $\pi.$
A sequence of local shortening operations, starting at a point that differs from $\pi$ and including $e$, decreases the slope of $e$ and decreases the height of the watchtower and $e$ did not determine an optimal watchtower.
\end{proof}

\begin{figure}[h]
\centering
\includegraphics[width=0.4\textwidth]{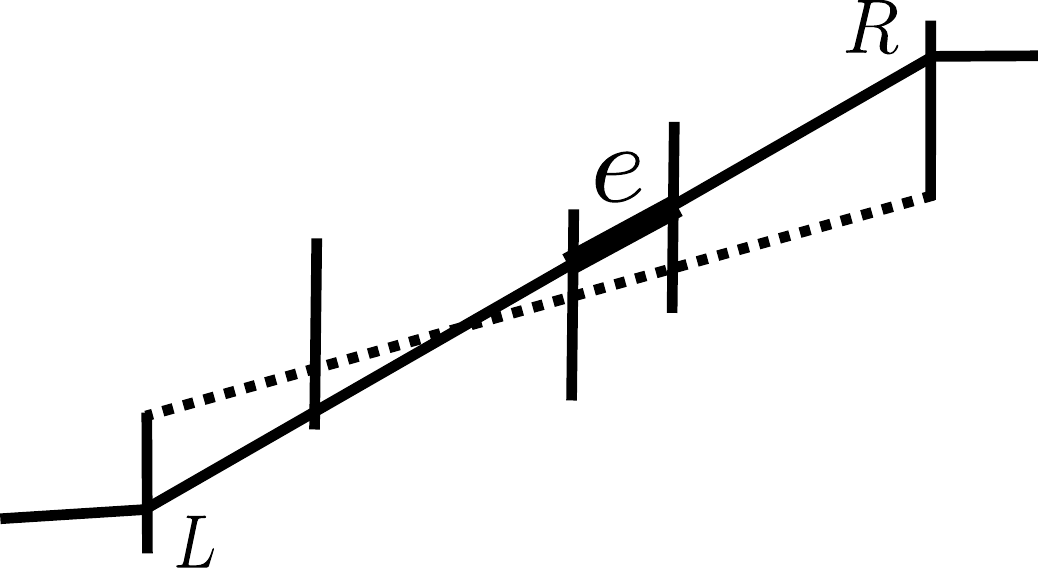}
\caption{The dashed line represents a segment of $\pi.$ Since $e$ is not on $\pi$ its slope is able to be decreased.}
\label{fig:optimal-edge}
\end{figure}

\subsection{Discrete Case}
We now consider the discrete case in which the base of the watchtower is located at a vertex of the terrain. 
In this case, raising the vertex that contains the base does not increase the height of the watchtower.

\begin{lemma}[Raising the base]\label{lem:base}
Let $T$ be a realization and let $W$ be a watchtower whose base lies on a vertex $v_i$.
Let $T'$ be the realization obtained from $T$ by raising $v_i$ to the top of its interval $t_i$.
Let $W'$ be the watchtower with the same top as $W$ and base at $t_i$.
Then $W'$ is a valid watchtower for $T'$.
\end{lemma}

\begin{proof}
When the base of the watchtower is at a vertex, the adjacent edges remain visible because the terrain is $x$-monotone.  
Since all other vertices of $T$ are unchanged, $W'$ continues to be a valid watchtower for $T'$.
\end{proof}

Therefore, in the discrete case, we may restrict our attention to at most $n$ realizations without increasing the height of the optimal watchtower.  
Specifically, we first compute a shortest-path realization from the top of $\ell_1$ to the top of $\ell_n$, 
and then, for each vertex in turn, raise it to the top of its interval while keeping all other vertices fixed.  
The resulting watchtower height is the vertical distance from the top of each interval to the visibility polygon of $\pi$, denoted by $P$.
This gives the following algorithm for the discrete case.
 \begin{algorithm}[h]
   \caption{Discrete Imprecise Guarding}\label{alg:base-on-interval}
    \begin{algorithmic}[1]

        \Require An imprecise terrain
        \Ensure Shortest discrete watchtower height

	\State Compute the visibility polygon of $\pi$, $P$\label{line:p-vis}
        \State Let $h=\infty$
        \For{Each interval $\ell_i$}\label{line:single-up}
        \State Compute $h'$ the vertical distance from the vertex at the top of the interval to $P$
	\If{$h'<h$}
	\State $h=h'$
	\EndIf
        \EndFor
      \State  \Return $h$
        
    \end{algorithmic}
\end{algorithm}

\begin{theorem}[Discrete Watchtower]\label{thm:tower-abover-vertex}
In the discrete case, the optimistic shortest watchtower can be computed in $O(n)$ time.
\end{theorem}
\begin{proof}

By \lemref{wings}, we may assume that $t_1$ and $t_n$ are vertices of an optimal realization.  
By \lemref{half-planes-on-pi}, the top of an optimal watchtower is determined by the extensions of edges along the shortest path $\pi : t_1 \to t_n$.  
Furthermore, \lemref{base} implies that the optimal watchtower height is equal to the minimum vertical distance from the top of an interval to the visibility polygon $P$ of $\pi$ and \algref{base-on-interval} is correct.

The polygon $P$ can be computed in $O(n)$ time using the algorithm of Lee and Preparata \cite{lee_optimal_1979}, as shown in \lnref{p-vis} of \algref{base-on-interval}.  
Since there are $n$ intervals to consider (see \lnref{single-up}), we evaluate the distance from each interval top to $P$.  
Since $P$ is the intersection of half-planes, it is convex, and its $O(n)$ vertices can be ordered from left to right in $O(n)$ time.  
Then, scanning $P$ from left to right allows us to compute the distance from the top of each interval to $P$ in constant time per interval.  
Thus, the overall computation takes $O(n)$ time.
\end{proof}

\subsection{Continuous Case}

We now consider the case where the base of the watchtower lies in the interior of an edge.  
Unlike the discrete case, raising the edge containing the base as high as possible does not necessarily minimize the watchtower height: doing so may steepen adjacent edges and increase the height of the watchtower. 
By \lemref{half-planes-on-pi}, we may assume the top of the watchtower is determined by the extensions of edges $e_i$ and $e_j$ of $\pi$.

Let $p$ denote the point on the visibility polygon $P$ of $\pi$ obtained as the intersection of the extended edges $e_i$ and $e_j$ (\figref{too-much-p}).  
The $x$-coordinate of $p$ lies between two consecutive intervals, denoted $\ell_{k}$ and $\ell_{k+1}$.
Recall, $Q$ is the simple polygon formed by $\ell_1$, the chain determined by connecting the consecutive tops of the intervals,
the chain formed by connecting the consecutive bottoms of the intervals, and $\ell_n.$
Define $Q_p$ as the polygon obtained by taking the union of $Q$ with the triangle $\triangle t_{k}pt_{k+1}$ (\figref{poly-p}).  

Let $\rho_1$ denote the shortest path from $t_1$ to $p$ in $Q_p$, and let $\rho_2$ denote the shortest path from $p$ to $t_n$ in $Q_p$ (\figref{too-much-rho}).  
Let $u_1$ be the rightmost vertex at which $\pi$ and $\rho_1$ coincide, and let $u_2$ be the leftmost vertex such that $\pi$ and $\rho_2$ coincide for all larger $x$-values.

\begin{figure}[t]
    \captionsetup[subfigure]{justification=centering}
    \centering
    \begin{subfigure}[t]{0.4\textwidth}
    	\includegraphics[width=\textwidth]{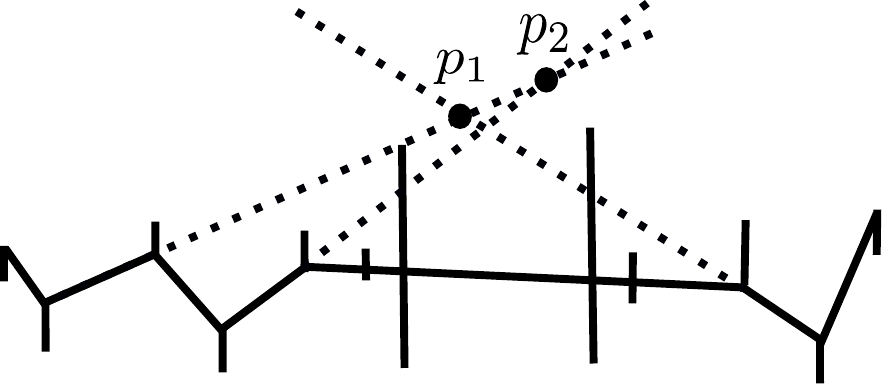}
	    \subcaption{}
    \label{fig:too-much-p}
    \end{subfigure}
    \hspace{.025cm}
    \begin{subfigure}[t]{0.42\textwidth}
  	    \includegraphics[width=\textwidth]{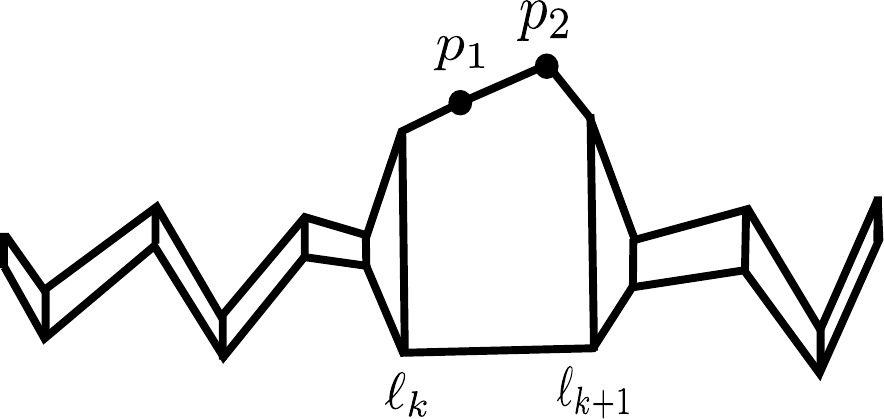}
	    \subcaption{}
    \label{fig:poly-p}
    \end{subfigure}

	\caption{(\subref{fig:too-much-p}) The intersection of extended edges determine the points $p_1$ and $p_2$. (\subref{fig:poly-p}) The polygon $Q_P$.}
\label{fig:intro}
\end{figure}

\begin{figure}[t]
    \captionsetup[subfigure]{justification=centering}
    \centering
    \begin{subfigure}[t]{0.42\textwidth}
    	\includegraphics[width=\textwidth]{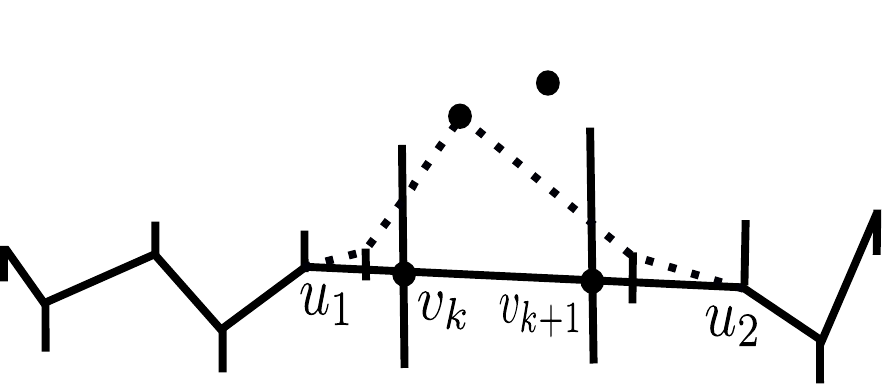}
	    \subcaption{}
    \label{fig:too-much-rho}
    \end{subfigure}
    \hspace{.25cm}
    \begin{subfigure}[t]{0.42\textwidth}
  	    \includegraphics[width=\textwidth]{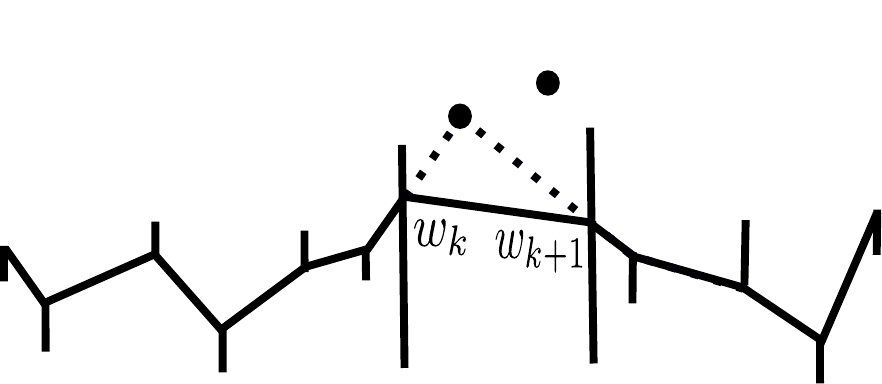}
	    \subcaption{}
    \label{fig:too-much-updated}
    \end{subfigure}

	\caption{(\subref{fig:too-much-rho}) The dashed lines indicate $\rho_1$ and $\rho_2$. (\subref{fig:too-much-updated}) We raise the edge $v_kv_{k+1}$ to $w_kw_{k+1}$ the intersections of the shortest paths in (\subref{fig:too-much-rho}) with $\ell_k$ and $\ell_{k+1}.$ This does not increase the height of the watchtower.}
\label{fig:intro}
\end{figure}

 \begin{lemma}[Convex $u_1$ to $p$]\label{lem:raise-to-p}
The subpath of $\rho_1$ from $u_1\rightarrow p$ only makes left turns.
\end{lemma}

\begin{proof}
If the path $\rho_1$ were to make a right turn, then it could be shortened, contradicting minimality. 
\end{proof}

Let $w_k$ and $w_{k+1}$ denote the points where $\rho_1$ and $\rho_2$ intersect $\ell_k$ and $\ell_{k+1}$, respectively.  
Define $\rho$ as the realization consisting of $\rho_1$ from $t_1$ to $w_k$, followed by the edge $w_kw_{k+1}$, and then $\rho_2$ from $w_{k+1}$ to $t_n$ (\figref{too-much-updated}).  
By \lemref{raise-to-p}, the watchtower height of $\rho$ is no greater than that of $\pi$.  

Now let $T$ be any realization with the watchtower base lying in the interior of an edge.  
We show that $T$ can be transformed into $O(n)$ realizations without increasing the watchtower height.  
Suppose the base of the watchtower lies on edge $e = v_{k}v_{k+1}$ of $T$.  
Let $T'$ denote the realization obtained from $T$ by raising $v_1$ and $v_n$ to the tops of their respective intervals and shortening the paths from $v_k \rightarrow t_1$ and from $v_{k+1} \rightarrow t_n$.  
By \lemref{wings} and \lemref{shorten}, the watchtower height of $T'$ is no greater than that of $T$.  

\begin{lemma}[Position Base Edge]\label{lem:position-base}
Moving $v_k$ to $w_k$ and $v_{k+1}$ to $w_{k+1}$ does not increase the height of a watchtower of $T$.
\end{lemma}

\begin{proof}
We first consider moving $v_k$ to $w_k$; the argument for $v_{k+1}$ is symmetric.  
Let $u'$ be the rightmost point where $T'$ and $\rho_1$ coincide.  
Since both $T'$ and $\rho_1$ are determined by shortest paths ending at $t_1$, the edge of $T'$ whose extension determines $p$ has $x$-value less than $u'$, and the $y$-coordinate of $v_k$ is less than that of $w_k$.  
Thus, raising $v_k$ to $w_k$ does not increase the watchtower height because, by \lemref{raise-to-p}, the path $u' \rightarrow w_k$ only makes left turns. 
\end{proof}


In the continuous case, we can morph any realization into one of $n-1$ canonical realizations without increasing the height of the optimal watchtower.  
For each vertex $p$ of the polygon $P$ associated with $\pi$,  
we compute the shortest paths from $p$ to $t_1$ and $t_n$ in $Q_p$, determine where these paths intersect the intervals immediately to the left and right of $p$, and construct the corresponding realization.  
Since both $p$ and the edge containing the base are known, the watchtower height can then be computed in constant time.

\begin{restatable}[Continuous Watchtower]{theorem}{continuouswatchtower}
\label{hm:tower-abover-edge}
In the continuous case,
the optimistic shortest watchtower can be computed in \run\  time.
\end{restatable}

\begin{proof}
As in the discrete case, by \lemref{wings} we assume that $t_1$ and $t_n$ are vertices of an optimal realization. By \lemref{half-planes-on-pi} the top of an optimal watchtower is determined by extensions of edges in $\pi.$
And \lemref{position-base} implies the base of the watchtower is raised as much as possible. Therefore, \algref{base-on-edge} is correct. 

We next justify the $O(n)$ run time.
There are $O(n)$ vertices of $P.$  
We define the simple polygon $\hat{Q}$ as follows.  
The upper chain consists of the tops of the intervals of the imprecise terrain together with the vertices of $P$, sorted by $x$-coordinate, with edges added between consecutive points in this order.  
The sides of $\hat{Q}$ are given by $\ell_1$ and $\ell_n$, while the bottom chain is formed by the bottoms of the intervals of the imprecise terrain.  

Using $\hat{Q}$, we compute shortest paths from $t_1$ and $t_n$ to all points of $P$ in linear time  
\cite{guibas_linear-time_1987}.  
The intersections of these paths with the adjacent intervals (\lnref{find-points} of \algref{base-on-edge}) are determined in constant time, since they are defined by the final edge of the path.  
Finally, in \lnref{tower-check} of \algref{base-on-edge}, the watchtower height is computed as the vertical distance from the base edge to $p$ in constant~time.  

\end{proof}

 \begin{algorithm}[h]
   \caption{Continuous Imprecise Guarding}\label{alg:base-on-edge}
    \begin{algorithmic}[1]

        \Require An imprecise terrain
        \Ensure The height of the shortest continuous watchtower
         \State Compute the visibility polygon of $\pi$, $P$
          \State Let $h=\infty$
         \State Compute the last edge of the shortest paths from $t_1$ and $t_n$ to vertices in $P$\label{line:shorties}
        \For{Each vertex $p$ in $P$ }\label{line:k-vertices}
        \State Let $\pi_k$ and $\pi_{k+1}$ denote the shortest paths from $t_1$
         and $t_n$ to $p$
        \State Find the points $v_k=\pi_k\cap \ell_k$ and $v_{k+1}=\pi_{k+1}\cap \ell_{k+1}$ \label{line:find-points}
        \State Let $h'$ be the distance from $p$ to edge $v_kv_{k+1}$ \label{line:tower-check}
	\If{$h'<h$}
	\State $h=h'$
	\EndIf
        \EndFor
        \State \Return $h$
        
    \end{algorithmic}
\end{algorithm}

\section{Algorithm for the Discrete Version in 2.5D}
\label{sec:ptas}

In this model, we start with an underlying planar triangulation of $n$ vertices $p_1,p_2,...,p_n$, with $p_i=(x_i,y_i)$; then each $\triangle p_ip_jp_k$ in
the plane is lifted so that each $p_i$ can possibly be within a vertical interval $I_i=((x_i,y_i,b_i),(x_i,y_i,t_i))$. We then use the 3D vertex (point) $q_i=(x_i,y_i,z_i)$ to represent a geometric realization of $p_i$, with $b_i\leq z_i \leq t_i$. Note that $\triangle p_ip_jp_k$ and $\triangle q_iq_jq_k$ are topologically equivalent. We call the lifted triangulation $T$ over 3D points $Q=\{q_i|i\in[n]\}$ an {\em imprecise terrain} over $Q$, keeping in mind that the $z$-coordinate of $q_i$ is a variable.

The optimistic shortest watchtower problem in 2.5D is then defined as
computing a shortest watchtower $|uv|$ such that $u$ is on $T$ and $v$ can see all points on $T$, among all possible realizations of $T$.

The problem is certainly harder than the 1.5D counterpart, one reason is probably that in 2.5D computing the shortest path between two vertices $q_i$ and $q_\ell$ on $T$ is already NP-hard \cite{gray-shortest-hard-04}. So, we need to first consider a more fundamental question:
is there an optimistic watchtower on $T$ of height zero? We call this problem the {\em zero-watchtower problem} on $T$.

It is easy to see that if a zero-watchtower solution exists, a solution must occur at some vertex $q_i$, i.e., $u=v=q_\ell$.

\begin{lemma}
If $q_\ell=(x_\ell,y_\ell,z_\ell)$ is a solution of the zero-watchtower problem on $T$, then $q'_\ell=(x_\ell,y_\ell,t_\ell)$ is also a solution.  
\label{3d-1}
\end{lemma}

\begin{proof}
 If $q_\ell$ ($\neq q'_\ell$) is a solution then for all other triangular face $F_j$'s in $T$ not adjacent to $q_\ell$, $q_\ell$ must be on $F_j$ or above the half-plane determined by $F_j$ --- otherwise, $q_\ell$ cannot guard some face $F_j$ and cannot be a solution. Fixing all these face $F_j$'s, we then lift $q_\ell$ vertically to $q'_\ell$. Clearly, $q'_\ell$ is above all $F_j$'s. The remaining faces are all adjacent to $q'_\ell$ (before the lift-up, adjacent to $q_\ell$), hence are all guarded by $q'_\ell$.
\end{proof}



Given an edge $e_k=p_ip_j$ on $T$, define $H(p_ip_j)$ or $H(e_k)$ as the vertical half-plane bounded above by the line through $e_k=p_ip_j$. We have the following~lemma.

\begin{lemma}
Fixing $q'_\ell=(x_\ell,y_\ell,t_\ell)$, there is no solution for the zero-watchtower problem at $q'_\ell$ if and only if (1) there is a vertex $p_k\in T$ whose highest realization $p'_k$
cannot see $q'_\ell$ regardless of the realization of other vertices; or (2) there is an edge $p_cp_d\in T$ whose lowest realization $p^{-}_cp^{-}_d$ is not visible to $q'_\ell$, regardless of the realization of other vertices.
\label{3d-2}
\end{lemma}

\begin{proof}
Note that in case (2), when $p^{-}_cp^{-}_d$ is not visible to $q'_\ell$, one of the two faces
adjacent to $p^{-}_cp^{-}_d$ is also not visible to $q'_\ell$.

We proceed to prove this ``if and only if'' relation. The ``only if'' part is easy to see. We focus on proving the ``if'' part, i.e., assuming that there is such a $p'_k$ or $p^{-}_cp^{-}_d$ satisfying the condition in Lemma~\ref{3d-2}. For Case (1),  when $p'_k$ and $q'_\ell$ are not visible to each other, there must be an edge $p_ip_j$ such that the segment $p'_kq'_\ell$ intersects
the vertical half-plane $H(p_ip_j)$; in fact, the segment would even intersect $H(p^{-}_ip^{-}_j)$. See \figref{hidden-point} for an example.
As a matter of fact, $q'_\ell$ cannot be a zero-watchtower solution. At this point, suppose there is no vertex $p_k$ satisfying the condition in Case (1). Now consider Case (2), assuming there is an edge $p_cp_d\in T$ whose lowest realization $p^{-}_cp^{-}_d$ is not visible to $q'_\ell$ and one of the two faces adjacent to it, say $\triangle p^{-}_cp^{-}_dp_f$, is not visible to $q'_\ell$ (but the three vertices $p^{-}_c,p^{-}_d$ and $p_f$ are all visible to $q'_\ell$). In this case, if we lower $p_f$ to its lowest realization $p^{-}_f$ and the edge $p^{-}_cp^{-}_d$ (and the corresponding face $\triangle p^{-}_cp^{-}_dp^{-}_f$) are still not visible to $q'_\ell$, then certainly $q'_\ell$ cannot be a zero-watchtower solution. See \figref{hidden-face} for an example.
\end{proof}

\begin{figure}[t]
    \captionsetup[subfigure]{justification=centering}
    \centering
    \begin{subfigure}[t]{0.37\textwidth}
    	\includegraphics[width=\textwidth]{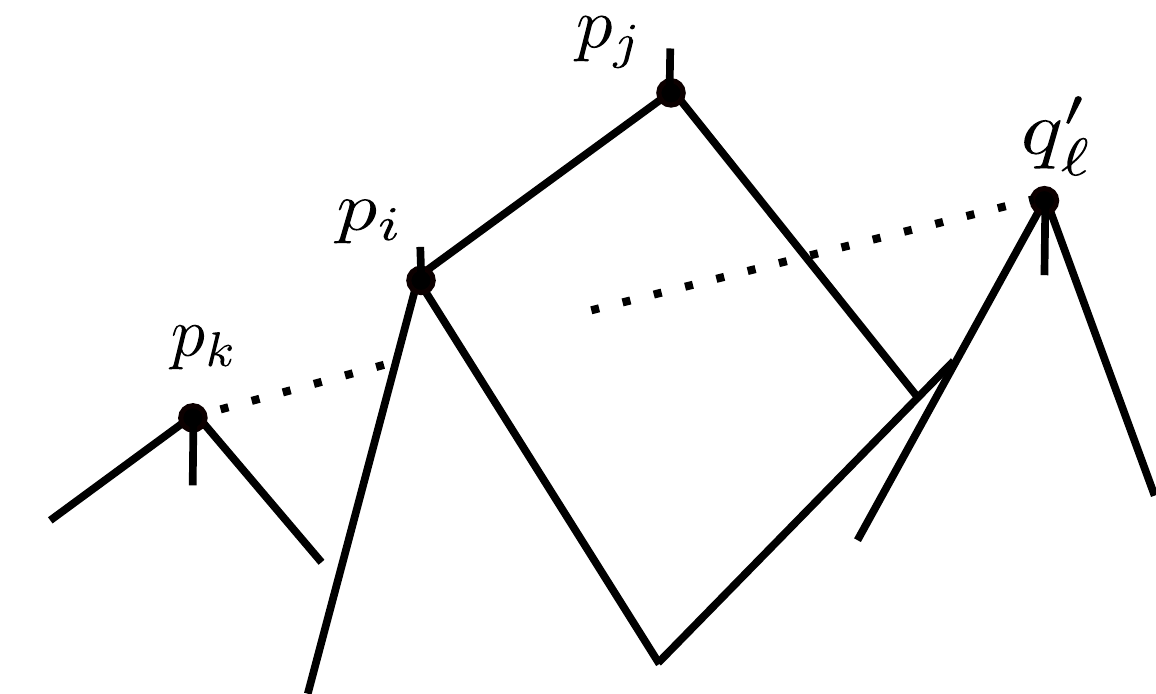}
	    \subcaption{}
    \label{fig:hidden-point}
    \end{subfigure}
    \hspace{.05cm}
    \begin{subfigure}[t]{0.42\textwidth}
  	    \includegraphics[width=\textwidth]{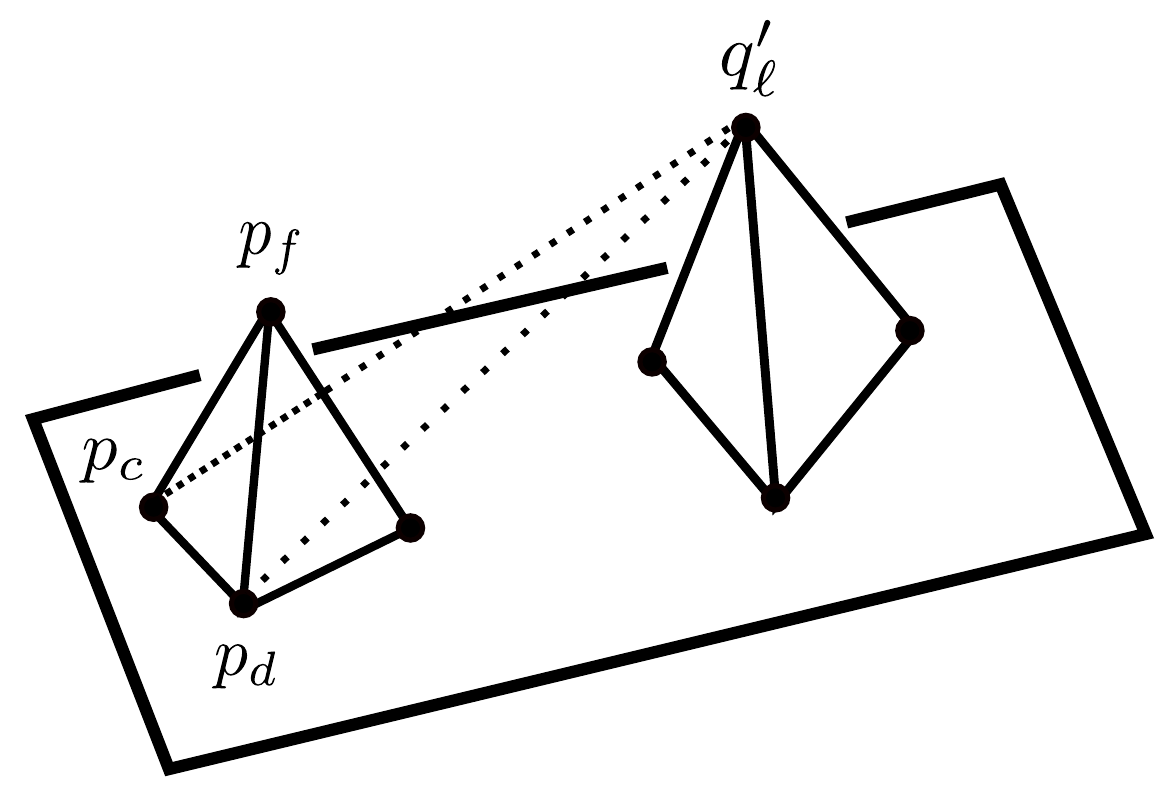}
	    \subcaption{}
    \label{fig:hidden-face}
    \end{subfigure}

	\caption{Illustration for the proof of Lemma~\ref{3d-2}.}
\label{fig:fig4}
\end{figure}


The proof of Lemma~\ref{3d-2} is constructive and it implies an algorithm to test if $q'_\ell$
is a solution for the zero-watchtower problem.

\begin{theorem}
 The zero-watchtower problem on $T$ can be decided in $O(n^3)$ time.
 \label{thm-3d}
\end{theorem}

\begin{proof} 
 With the above lemma, we obtain an algorithm as follows. We start with the highest realization of 
 all vertices in $T$ and then focus on $q'_\ell$. If $q'_\ell$ can see all the edges of the
 current realization of $T$, then we are done. If not, following Lemma~\ref{3d-2}, we have the following two~cases.
 
 In Case (1),  there is a $p'_k$ (the highest realization of $p_k$) such that $p'_k$ and $q'_\ell$ cannot see each other, i.e., they intersect some $H(p_ip_j)$. We label $p'_k$ as a candidate vertex,
 and identify all the edges of $T$, say $e_{k,1},e_{k,2},...,e_{k,s_k}$, whose vertical half-planes $H(e_{k,1}), H(e_{k,2}),\cdots, H(e_{k,s_k})$ intersect the line segment $p'_{k}q'_\ell$. (If there are multiple such candidate vertices, we process the farthest one from $q'_\ell$ first.)  We then lower the realizations
 of these $s_k$ edges by translating them vertically down to a position such that $p'_k$ and $q'_\ell$ are barely visible to each other (i.e., they just touch the line segment $p'_{k}q'_\ell$). If it is impossible to make $p'_k$ and $q'_\ell$ visible by these downward translations, then we conclude that there is no solution for the zero-watchtower problem at $q'_\ell$.
 Otherwise, we repeat the above process until all candidate vertex $p'_k$'s are now visible to $q'_\ell$ but $q'_\ell$ still cannot cover the current realization of $T$, then we continue with the next step.
 
 By now in Case (2), suppose there is an edge $p_cp_d$ whose lowest realization $p^{-}_cp^{-}_d$
 cannot be seen from $q'_\ell$. As discussed earlier, one of the two faces adjacent to
 $p^{-}_cp^{-}_d$, say $\triangle p^{-}_cp^{-}_dp_f$, is also not visible to $q'_\ell$. We then lower $p_f$ to its lowest realization $p^{-}_f$. If $\triangle p^{-}_cp^{-}_dp^{-}_f$ is visible to
 $q'_\ell$ then we repeat this process for Case (2) to find  and process some other edges not visible to $q'_\ell$; otherwise, we report that $q'_\ell$ cannot be a solution for the zero-watchtower problem.
 
 Note that in the whole process for Case (1) and Case (2), except at the beginning step, we only lower vertices and edges to their lowest realization, hence there is no conflict. Each step of Case (1) and Case (2) obviously takes $O(n)$ time, and in each case we could have $O(n)$ steps (e.g., in Case (1) we have $O(n)$ candidate vertices $p'_k$'s). Clearly, if there is a solution when $q'_\ell$ fixed, we would spend $O(n^2)$ time to return a solution.

 Since we have to try $O(n)$ vertex $q'_\ell$'s, the whole algorithm takes $O(n^3)$ time to either return a solution, or report that no solution exists for the zero-watchtower~problem.
\end{proof}

Theorem~\ref{thm-3d} also implies that we could obtain an additive approximation scheme for the discrete shortest watchtower problem on $T$ where the base must be at a vertex of $T$. The algorithm first decides if
a zero-solution exists; if so then we can solve the problem in $O(n^3)$ time. Otherwise,
at each candidate base vertex $q'_\ell$ we check the point $q''_\ell$  which is $\varepsilon$ distance above $q'_\ell$ vertically. If $q''_\ell$ can see every edge of a realization of $T$ then
return the vertical segment $q'_\ell q''_\ell$ as the approximate solution. Otherwise, update $q''_\ell \leftarrow q''_\ell \oplus\varepsilon$ and repeat the algorithm in Theorem~\ref{thm-3d} (where $v\oplus \varepsilon$ is the operation to move $v$ by a vertical distance $\varepsilon$). If the optimal solution value is ${OPT}$, then with the linear search method we can certainly find a solution with solution value at most ${OPT}+\varepsilon$ in $O(\frac{{OPT}}{\varepsilon}n^3)$ time. Note that with the exponential search method we could improve the time bound to $O(\log (\frac{{OPT}}{\varepsilon})n^3)$.

\begin{corollary}
An additive approximation scheme for the discrete shortest watchtower on an imprecise terrain $T$ with $n$ vertices
can be computed in\\ $O(\log (\frac{{OPT}}{\varepsilon})n^3)$ time, and with a solution value at most ${OPT}+\varepsilon$, where ${OPT}$ is the optimal solution value of the problem.
\end{corollary}

\section{Concluding Remarks}
\label{sec:discuss}
In this paper we solve the optimistic shortest watchtower problem optimally for the 1.5D case. For the discrete version in 2.5D we give an additive approximation scheme which runs in $O(\log (\frac{{OPT}}{\varepsilon})n^3)$ time, where $n$ is the number of imprecise vertical intervals and ${OPT}$ is the optimal solution value. A natural question is to determine the computational complexity of the 2.5D problem. Given that finding shortest paths is hard in this setting,
we conjecture that the 2.5D problem is intractable. As a further evidence for our conjecture, we comment that for the discrete version in 2.5D, our approximation scheme is not a PTAS, as ${OPT}$ is not a part of the input.
Hence, designing a PTAS for the discrete version is itself an interesting open~problem.

\section*{Acknowledgments}

This research is partially supported by NSF under project CNS-2243010 and CCF-2529539. We also thank anonymous reviewers for their comments.
\bibliographystyle{plainurl}
\bibliography{references}

@inproceedings{driemel_flow_2011,
	title = {Flow Computations on Imprecise Terrains},
	booktitle = {Algorithms and {Data} {Structures} (WADS)},
	publisher = {Springer},
	author = {Driemel, Anne and Haverkort, Herman and Löffler, Maarten and Silveira, Rodrigo I.},
	year = {2011},
	pages = {350--361},
}

@article{CS89,
  author       = {Richard Cole and
                  Micha Sharir},
  title        = {Visibility Problems for Polyhedral Terrains},
  journal      = {J. Symb. Comput.},
  volume       = {7},
  number       = {1},
  pages        = {11--30},
  year         = {1989}
}

@article{king_terrain_2011,
	title = {Terrain Guarding is {NP}-Hard},
	volume = {40},
	number = {5},
	journal = {SIAM Journal on Computing},
	author = {King, James and Krohn, Erik},
	year = {2011},
	pages = {1316--1339},
}

@article{Loffler_viewpoints_2025,
    author = {Keikha, Vahideh and Löffler, Maarten and Saumell, Maria and Valtr, Pavel},
    title = {Guarding a 1.5{D} terrain with Imprecise Viewpoints},
    journal = {European Workshop on Computational Geometry (EuroCG)},
    year = {2025}
}

@article{guibas_linear-time_1987,
	title = {Linear-time algorithms for visibility and shortest path problems inside triangulated simple polygons},
	volume = {2},
	number = {1},
	journal = {Algorithmica},
	author = {Guibas, Leonidas and Hershberger, John and Leven, Daniel and Sharir, Micha and Tarjan, Robert E.},
	year = {1987},
	pages = {209--233},
}

@article{seth_acrophobic_2023,
	title = {Acrophobic guard watchtower problem},
	volume = {109},
	journal = {Computational Geometry},
	author = {Seth, Ritesh and Maheshwari, Anil and Nandy, Subhas C.},
	month = feb,
	year = {2023},
	pages = {101918},
}

@article{benmoshe_constantfactor_2007,
	title = {A Constant-Factor Approximation Algorithm for Optimal 1.5{D} Terrain Guarding},
	volume = {36},
	language = {en},
	number = {6},
	journal = {SIAM Journal on Computing},
	author = {Ben-Moshe, Boaz and Katz, Matthew J. and Mitchell, Joseph S. B.},
	year = {2007},
	pages = {1631-1647},
}

@incollection{chen_guarding_2018,
	address = {Cham},
	title = {Guarding Polyhedral Terrain by k-Watchtowers},
	volume = {10823},
	urldate = {2025-01-03},
	booktitle = {Frontiers in {Algorithmics}},
	publisher = {Springer International Publishing},
	author = {Tripathi, Nitesh and Pal, Manjish and De, Minati and Das, Gautam and Nandy, Subhas C.},
	year = {2018},
	pages = {112--125},
}

@article{agarwal_guarding_2010,
	title = {Guarding a Terrain by Two Watchtowers},
	volume = {58},
	number = {2},
	journal = {Algorithmica},
	author = {Agarwal, Pankaj K. and Bereg, Sergey and Daescu, Ovidiu and Kaplan, Haim and Ntafos, Simeon and Sharir, Micha and Zhu, Binhai},
	year = {2010},
	pages = {352--390},
}

@article{sharir_shortest_1988,
	title = {The shortest watchtower and related problems for polyhedral terrains},
	volume = {29},
	number = {5},
	journal = {Information Processing Letters},
	author = {Sharir, Micha},
	year = {1988},
	pages = {265--270},
}

@inproceedings{gray-shortest-hard-04,
  author    = {Gray, Chris and Evans, William},
  title     = {Optimistic shortest paths on uncertain terrains},
  booktitle = {Proceedings of the 16th Canadian Conference on Computational Geometry,
               (CCCG)},
  pages     = {68--71},
  year      = {2004}
}

@article{lee_optimal_1979,
	title = {An Optimal Algorithm for Finding the Kernel of a Polygon},
	volume = {26},
	number = {3},
	journal = {Journal of the ACM},
	author = {Lee, D. T. and Preparata, F. P.},
	year = {1979},
	pages = {415--421},
}

@inproceedings{mccoy_guarding_2024,
	title = {Guarding Precise and Imprecise Polyhedral Terrains with Segments},
	booktitle = {Combinatorial {Optimization} and {Applications} (COCOA)},
	publisher = {Springer Nature Switzerland},
	author = {McCoy, Bradley and Zhu, Binhai and Dutt, Aakash},
	year = {2023},
	pages = {323-336},
}

@article{Lubiw_Stroud_2023, 
title={Computing Realistic Terrains from Imprecise Elevations}, 
volume={2},  
number={2}, 
journal={Computing in Geometry and Topology}, 
author={Lubiw, Anna and Stroud, Graeme}, 
year={2023},
 pages={3:1–3:18} 
 }

@article{zhu_computing_1997,
	title = {Computing the shortest watchtower of a polyhedral terrain in ${O}(n\log n)$ time},
	volume = {8},
	number = {4},
	urldate = {2023-02-18},
	journal = {Computational Geometry},
	author = {Zhu, Binhai},
	year = {1997},
	pages = {181--193},
}

@article{gray_smoothing_2010,
	title = {Smoothing Imprecise 1.5{D} Terrains},
	volume = {20},
	number = {04},
	urldate = {2023-02-18},
	journal = {International Journal of Computational Geometry \& Applications},
	author = {Gray, Chris and Löffler, Maarten and Silveira, Rodrigo I.},
	year = {2010},
	pages = {381--414},
}

@article{gray_removing_2012,
	title = {Removing local extrema from imprecise terrains},
	volume = {45},
	number = {7},
	urldate = {2023-02-18},
	journal = {Computational Geometry},
	author = {Gray, Chris and Kammer, Frank and Löffler, Maarten and Silveira, Rodrigo I.},
	year = {2012},
	pages = {334--349},
}

\end{document}